\documentclass[10pt,conference,twocolumn,a4paper]{IEEEtran}
\usepackage{amsfonts,amsmath,amssymb,amstext,amsthm,color,dsfont,enumerate,hyperref,verbatim,xfrac,xr,url}
\usepackage[mathscr]{euscript}
\hypersetup{
pdftitle		={The Augustin Center and The Sphere Packing Bound For Memoryless Channels},
pdfauthor		={Baris Nakiboglu},
pdfsubject		={The Cost Constrained Augustin Capacity},
pdfkeywords		={-}, 
}
\title{\vspace{-.2cm}The Augustin Center and The Sphere Packing 
\\\vspace{-.15cm} Bound For Memoryless Channels\vspace{-.3cm}}
\author{Bar\i\c{s} Nakibo\u{g}lu\\\!\vspace{-.3cm}\small{\href{mailto:nakib@alum.mit.edu}{nakib@alum.mit.edu}}
\vspace{-.1cm}} 
\theoremstyle{plain}
\newtheorem{lemma}{Lemma} 
\newtheorem{theorem}{Theorem}

\theoremstyle{definition}
\newtheorem{definition}{Definition} 
\newtheorem{assumption}{Assumption}

\theoremstyle{remark}

\newcommand{\set} [1]			{{\mathscr{{#1}}}}
\newcommand{\alg}[1]			{{\mathcal{{#1}}}}

\newcommand{\oper}[1]      {{\mathtt{{#1}}}}
\newcommand{\msr}[1]       {{\it    {{#1}}}}
\newcommand{\cnst}[1]      {{\mathit{{#1}}}}

\newcommand{\cntt}[1]      {{\widetilde{{\mathit{{#1}}}}}}

\newcommand{\integers}[1]	{{\mathbb{Z}}_{^{{#1}}}}

\newcommand{\reals}[1]		{{\Re}_{^{{#1}}}}
\newcommand{\bigo} [1]     {{\cnst{O}({{#1}})}}

\newcommand{\inte}[1]      {{\mathtt{int}{{#1}}}}

\newcommand{\conv}[1]      {{\mathtt{ch}{{#1}}}}

\newcommand{\dif}[1]       {{\mathrm{d}{#1}}}  
\newcommand{\der}[2]        {\tfrac{\dif{#1}}{\dif{#2}}}

\newcommand{\DEF}[0]           {{\triangleq}}  
 
\newcommand{\oprod}[1]         {{\prod\nolimits_{{#1}}^{\otimes}}} 
\newcommand{\cprod}[1]         {{\prod\nolimits_{{#1}}^{\times}}}

\newcommand{\AC}[0]            {{\prec}}

\newcommand{\abs}[1]           {{\left\lvert{{#1}}\right\lvert}}

\newcommand{\lon}[1]           {{{\left\lVert{{#1}}\right\lVert}}} 

\newcommand{\ind}[0]           {{\imath}}
\newcommand{\jnd}[0]           {{\jmath}}
\newcommand{\knd}[0]           {{\kappa}}
\newcommand{\lnd}[0]           {{\ell}}
\newcommand{\tin}[0]           {{\cnst{t}}}
\newcommand{\blx}[0]           {{\cnst{n}}}

\newcommand{\tinS}[0]          {{\set{T}}}

\newcommand{\EXS}[2]         {{\bf E}_{{#1}}\!\left[{#2}\right]}

\newcommand{\fX}[0]          {{\cnst{f}}}

\newcommand{\cm}[1]         {{{\cnst{g}}_{{#1}}}}

\newcommand{\RD}[3]				{{\cnst{D}}_{{#1}}            \!\left(\left.            \! {#2}\right\Vert {#3}                  \right)}

\newcommand{\CRD}[4]			{{\cnst{D}}_{{#1}}            \!\left(\left.\!\left.    \! {#2}\right\Vert {#3} \right\vert{{#4}}\right)}

\newcommand{\RMI}[3]			{{\cnst{I}}_{{#1}}            \!\left(                  \! {#2};         \!{#3}                \!\right)}

\newcommand{\RMIL}[4]			{{\cnst{I}}_{{#1}}^{{#4}}	  \!\left(                  \! {#2};         \!{#3}                \!\right)} 
\newcommand{\GMIL}[4]			{{\cnst{I}}_{{#1}}^{{\scriptscriptstyle g}{#4}}	  \!\left(                  \! {#2};         \!{#3}                \!\right)} 
 
\newcommand{\GCL}[3]			{{\cnst{C}}_{{#1},{#2}}^{{\scriptscriptstyle g}{#3}}}
\newcommand{\RCL}[3]			{{\cnst{C}}_{{#1},{#2}}^{{#3}}}
\newcommand{\RRL}[3]			{{\cnst{S}}_{{#1},{#2}}^{{#3}}}

\newcommand{\CRC}[3]			{{\cnst{C}}_{{#1},{#2},{#3}}}
\newcommand{\CRCI}[4]			{{\cntt{C}}_{{#1},{#2},{#3}}^{#4}}

\newcommand{\Aopi}[4]			{{\oper{T}}_{{#1},{#2}}^{#3}({#4})} 
\newcommand{\Aop}[3]			{{\Aopi{#1}{#2}{}{#3}}} 
\newcommand{\cf}[0]				{{\cnst{\rho}}}
\newcommand{\cc}[0]				{{\cnst{\varrho}}}
\newcommand{\uc}[0]				{{\mathds{1}}}
\newcommand{\fcc}[1]			{{\cnst{\Gamma}_{{#1}}}}

\newcommand{\lm}[0]				{{\cnst{\lambda}}}

\newcommand{\cset}[0]			{{\set{A}}}

\newcommand{\spa}[2]			{{\cntt{E}_{sp\!}^{{#1}}}\left({#2}\right)}
\newcommand{\spe}[1]			{{\cnst{E}_{sp\!}}       \left({#1}\right)}
\newcommand{\rate}[0]			{{\cnst{R}}}
\newcommand{\cla}[2]          {{{\xi}_{{#1}}^{{#2}}}}

\newcommand{\rnf}[0]          {{\cnst{\phi}}}

\newcommand{\rno}[0]          {{\cnst{\alpha}}}
\newcommand{\rnt}[0]          {{\cnst{\eta}}}
\newcommand{\Pem}[1]           {{\it P_{{{\bf e}}}^{{#1}}}}         
\newcommand{\enc}[0]           {{\varPsi}} 
\newcommand{\dec}[0]           {{\varTheta}}    

\newcommand{\brl}[0]           {{\alg{B}}}
\newcommand{\rborel}[1]        {{\brl}({#1})}

\newcommand{\oev}[0]           {{\set{E}}}

\newcommand{\fmea}[1]          {{{\alg{M}}^{^{+}}\!({#1})}}

\newcommand{\pmea}[1]          {{{\alg{P}}({#1})}}
\newcommand{\fdis}[1]          {{{\set{M}}^{^{+}}\!({#1})}}
\newcommand{\pdis}[1]          {{{\set{P}}({#1})}}

\newcommand{\dinp}[0]          {{\cnst{x}}}

\newcommand{\inpS}[0]          {{\set{X}}}

\newcommand{\dout}[0]          {{\cnst{y}}}

\newcommand{\outS}[0]          {{\set{Y}}}
\newcommand{\outA}[0]          {{\alg{Y}}}

\newcommand{\dmes}[0]          {{\cnst{m}}}

\newcommand{\mesS}[0]          {{\set{M}}}

\newcommand{\dest}[0]          {{\widehat{{\cnst{m}}}}}

\newcommand{\estS}[0]          {{\widehat{{\set{M}}}}}

\newcommand{\mean}[0]        {{{\msr{\mu}}}}    
\newcommand{\mmn}[1]         {{{\mean}_{{#1}}}}
\newcommand{\mma}[2]         {{{\mean}_{{#1}}^{{#2}}}}    
\newcommand{\qgn}[1]         {{{\mQ}_{{#1}}^{{\scriptscriptstyle g}}}}
\newcommand{\qga}[2]         {{{\mQ}_{{#1}}^{{\scriptscriptstyle g}{#2}}}}

\newcommand{\rfm}[0]         {{{\msr{\nu}}}}

\newcommand{\mP}[0]				{{\msr{p}}}    
\newcommand{\pmn}[1]			{{{\mP}_{{#1}}}}
\newcommand{\pma}[2]			{{{\mP}_{{#1}}^{{#2}}}}

\newcommand{\mQ}[0]				{{\msr{q}}}    
\newcommand{\qmn}[1]			{{{\mQ}_{{#1}}}}
\newcommand{\qma}[2]			{{{\mQ}_{{#1}}^{{#2}}}}

\newcommand{\mU}[0]				{{\msr{u}}}

\newcommand{\mV}[0]				{{\msr{v}}}    

\newcommand{\vma}[2]			{{{\mV}_{{#1}}^{{#2}}}}

\newcommand{\mW}[0]				{{\msr{w}}}    
\newcommand{\wmn}[1]			{{{\mW}_{{#1}}}}

\newcommand{\Wm}[0]				{{{\cnst{W}}}}
\newcommand{\Wmn}[1]			{{{\cnst{W}}_{{#1}}}}

\begin{document}
\pagestyle{plain}
%\pagenumbering{arabic}
\maketitle 
\begin{abstract}
For any channel with a convex constraint set and finite Augustin capacity,
existence of a unique Augustin center and associated Erven-Harremoes bound 
are established. 
Augustin-Legendre capacity, center, and radius are introduced and proved to be equal 
to the corresponding Renyi-Gallager entities. 
Sphere packing bounds with polynomial prefactors are derived for codes on two families 
of channels:
(possibly non-stationary) memoryless channels with multiple additive cost constraints  
and 
stationary memoryless channels with convex constraints on the empirical distribution 
of the input codewords.
\end{abstract}

\section{Introduction}\label{sec:introduction}
Augustin \cite{augustin69}, \cite{augustin78} derived the sphere packing bound 
for the product channels without assuming the stationarity.
Assuming that order \textonehalf\!~Renyi capacity of the component channels are \(\bigo{\!\ln\!\blx\!}\),
we have derived the sphere packing bound for product channels with a prefactor that is polynomial 
in the block length \(\blx\), \cite[Theorem \ref*{B-thm:productexponent}]{nakiboglu16B}.
In this manuscript, we derive analogous results for two families of memoryless channels. 
As we have done for the product channels in \cite{nakiboglu16B}, we first derive a non-asymptotic 
outer bound for codes on a given memoryless channel, then we derive our asymptotic result using 
this bound.

In \cite[Chapter VII]{augustin78}, Augustin pursued an analysis similar to ours and derived the 
sphere packing bound for memoryless channels with cost constraints \cite[\S36]{augustin78}.
In addition, Augustin established the connection between the exponent of Gallager's inner bound for the cost 
constrained channels \cite[Thm 8]{gallager65} and the sphere packing exponent \cite[\S35]{augustin78}.
Our results surpass Augustin's results in two ways:
\begin{itemize}
\item Augustin assumes the cost function to be 
bounded.\footnote{The issue here is not a matter of rescaling: certain conclusions of 
Augustin's analysis are not correct when cost functions are not bounded.} 
This hypothesis excludes certain important and interesting cases such as the Gaussian 
channels. 
Hence, Augustin's analysis in \cite{augustin78} does not imply the sphere packing bounds
derived by Shannon \cite{shannon59} and Ebert \cite{ebert66}.
We don't assume the cost function to be bounded. Thus, Theorem \ref{thm:CCexponent} establishes
the sphere packing bound for a wider class of channels including the Gaussian channels with 
multiple antennas.\footnote{Shannon's approximation error terms in \cite{shannon59} are considerably 
better than ours. But his derivation relies heavily on the geometry of the output space.
Our derivation, on the other hand, is oblivious towards it.}
It is even possible to handle certain fading scenarios and additional per antenna power 
constraints.
%\footnote{When analyzing a specific family of channels, one often derives explicit
%expression for the sphere packing exponent. Derivation of such expressions, however, is beyond 
%the scope of the current manuscript.}

\item The best asymptotic bound implied by Augustin's non-asymptotic bound
\cite[Thm 36.6]{augustin78} is of the form
\(\Pem{av}^{(\blx)}\geq \bigo{\tfrac{1}{e^{\sqrt{\blx}}}}e^{-\spe{\ln\frac{M_{\blx}}{L_{\blx}}-\bigo{\sqrt{\blx}},\Wmn{[1,\blx]},\cc_{\blx}}}\).
In Theorem \ref{thm:CCexponent} we replace 
\(\bigo{\tfrac{1}{e^{\sqrt{\blx}}}}\) by \(\bigo{\tfrac{1}{\blx^{\tau}}}\) by \(\bigo{\sqrt{\blx}}\) to \(0\).
\end{itemize}

For stationary memoryless channels with finite input sets, 
the sphere packing bound is well-known \cite[Ch. 10]{csiszarkorner}, \cite{dalai16}.
For such a channel, one first chooses the most populous constant composition sub-code and then 
derives the sphere packing bound for the code using the sphere packing bound for the constant 
composition sub-code.\footnote{Haroutunian \cite{haroutunian68} was the first one to give a 
complete proof of the sphere packing bound for constant composition codes.
Recently, Altug and Wagner \cite{altugW14A} sharpened the prefactor of the bound for channels with
finite output sets. } 
This technique, however, fails when the input set of the channel is infinite. 
We show that a sphere packing bound similar to Theorem \ref{thm:CCexponent} holds 
for codes on stationary memoryless channels with convex constraints on the empirical distribution 
of the input codewords.

%%%%%Recently, Vazquez-Vilar, Martinez, and Fabregas proposed a derivation of the sphere 
%%%%%packing bound for stationary channels with a single cost constraint \cite{vazquezMF15}.
%%%%%We believe, their proof has a non-trivial gap. 
%%%%%In particular, they assert that \(Q^{n}\) does not depend on \({\bf x}_{m}\) in 
%%%%%\cite[eq.(26)]{vazquezMF15}.
%%%%%In order to assert that one has to include a supremum over \({\bf x}_{m}\) as the inner 
%%%%%most optimization both in \cite[eq.(25) \& eq.(26)]{vazquezMF15}. 
%%%%%With the additional supremum explanation given on \cite[p 931]{vazquezMF15} is no longer 
%%%%%valid.
%%%%%A similar problem is known to exist with the proof of the sphere packing bound proposed 
%%%%%by Blahut in \cite{blahut74}, see \cite[p1594 \& Appendix A]{altugW14A}.
%%%%%We have confirmed Altug and Wagner's observation analytically in \cite{nakiboglu16C}.
%%%%%In short Blahut's proof can not be salvaged without invoking a constant composition 
%%%%%argument. 
%%%%%Without a constant composition argument analysis presented by Blahut in \cite{blahut74}
%%%%%and Vazquez-Vilar, Martinez, and Fabregas in \cite{vazquezMF15} establish outer bounds
%%%%%with exponents strictly larger than the sphere packing exponent.
%%%%%Hence neither Blahut's paper \cite{blahut74} nor \cite{vazquezMF15} prove the sphere
%%%%%packing bound for memoryless channels with arbitrary input sets.

In the rest of this section, we describe our model and notation
and state our main asymptotic result.
In Section \ref{sec:center}, we introduce and analyze Augustin information, mean, 
capacity, and center as purely measure theoretic concepts. 
The role of these concepts in our analysis is analogous to the role of corresponding 
Renyi concepts in \cite{nakiboglu16A}, \cite{nakiboglu16B}.
In Section \ref{sec:costconstrained}, we investigate the cost constrained Augustin 
capacity more closely and introduce the concepts of Augustin-Legendre information 
and Renyi-Gallager information, together with the associated means, capacities, 
centers, and radii.
Our main aim in Section \ref{sec:costconstrained} is to express the cost constrained 
Augustin capacity and center in terms of Augustin-Legendre capacity and center. 
In Section \ref{sec:outerbound}, we derive non-asymptotic outer bounds for 
codes on two families memoryless channels.

\subsection{Model and Notation}\label{sec:model}
For any set \(\inpS\), \(\pdis{\inpS}\) is the set of all probability mass functions 
that are non-zero only on finitely many members of \(\inpS\);
\(\fdis{\inpS}\) is the set of all non-zero mass functions with the same property.
For any measurable space \((\outS,\outA)\), 
\(\pmea{\outA}\) is the set of all probability measures
and \(\fmea{\outA}\) is set of all finite measures.
For any \(\mean,\mQ\in\fmea{\outA}\),
\(\mean\leq \mQ\) iff \(\mean(\oev)\leq \mQ(\oev)\)  \(\forall\oev \in \outA\). 
Similarly, for any \(\mean,\mQ\in\reals{}^{\lnd}\), \(\mean\leq \mQ\) iff \(\mean^{\ind}\leq\mQ^{\ind}\) 
\(\forall\ind \in \{1,\ldots,\lnd\}\). 
For any \(\mean,\mQ\in\reals{}^{\lnd}\),
\(\mean \cdot \mQ\DEF\sum_{\jnd=1}^{\lnd} \mean^{\jnd} \mQ^{\jnd}\).
For any \(\ell\in\integers{+}\), \(\uc\in\reals{}^{\lnd}\) is the vector 
whose all entries are one.
For any \(\set{S}\subset \reals{}^{\lnd}\) 
we denote the interior of \(\set{S}\) by \(\inte{\set{S}}\).
For any set \(\set{S}\) in a vector space
we denote the convex hull of \(\set{S}\) by \(\conv{\set{S}}\).

A \emph{channel} \(\Wm\) is a function from \emph{the input set} \(\inpS\) to  the set of all probability 
measures on \emph{the output space} \((\outS,\outA)\).
A channel \(\Wm:\inpS\to\pmea{\outA}\) is a \emph{product channel} for a finite index set \(\tinS\)
iff there exist channels \(\Wmn{\tin}:\inpS_{\tin}\to\pmea{\outA_{\tin}}\) for all \(\tin\in\tinS\)
satisfying \(\Wm(\dinp)=\oprod{{\tin\in\tinS}}\Wmn{\tin}(\dinp_{\tin})\)
for all \(\dinp\in \inpS\)
where 
\begin{align}
\vspace{-.1cm}
\notag
\inpS
&= \oprod{{\tin\in\tinS}}\inpS_{\tin}
&
\outS
&= \cprod{{\tin\in\tinS}}\outS_{\tin}
&
\outA
&= \oprod{{\tin\in\tinS}}\outA_{\tin}.
\vspace{-.1cm}
\end{align}
A product channel is \emph{stationary} iff all \(\Wmn{\tin}\)'s are identical.
If \(\inpS\subset \oprod{{\tin\in\tinS}}\inpS_{\tin}\) then \(\Wm\) is a 
\emph{memoryless channel}.

An \((M,L)\) \emph{channel code} on \(\Wm:\inpS\to\pmea{\outA}\) 
is an ordered pair \((\enc,\dec)\) composed of  an \emph{encoding function} \(\enc:\mesS\to\inpS\) 
and a \emph{decoding function}\footnote{Recall that for any encoder \(\enc\) a deterministic MAP 
decoder obtains minimum \(\Pem{av}\) among all, possibly non-deterministic, decoders.}
\(\dec:\outS\to\estS\) where \(\mesS\DEF\{1,2,\ldots,M\}\), 
\(\estS\DEF\{\set{L}:\set{L}\subset\mesS \mbox{~and~}\abs{\set{L}}= L\}\),
and \(\dec\) is a  measurable as a function from 
the measurable space \((\outS,\outA)\).

Given an \((M,L)\) channel code \((\enc,\dec)\) on \(\Wm:\inpS\to\pmea{\outA}\) 
\emph{the average error probability} \(\Pem{av}\) and 
\emph{the conditional error probability} \(\Pem{\dmes}\) for \(\dmes \in \mesS\) 
are given by
\begin{align}
\notag
\Pem{av} 
&\DEF\tfrac{1}{M} \sum\nolimits_{\dmes\in \mesS} \Pem{\dmes}
&&&&&
\Pem{\dmes}
&\DEF \Wm(\enc(\dmes))(\{\dmes\notin\dest\}).
\end{align}

A \emph{cost function} \(\cf\) is a function from the input set  to \(\reals{\geq0}^{\lnd}\) for 
some \(\ell\in\integers{+}\).
We assume without loss of generality that\footnote{Augustin \cite[\S33]{augustin78} has the following additional 
hypothesis: \(\vee_{\dinp\in \inpS}\cf(\dinp)\!\leq\!\uc\).} 
\begin{align}
\notag
\inf\nolimits_{\dinp\in \inpS} \cf^{\ind}(\dinp) 
&=0
&
&\forall \ind \in \{1,\ldots,\lnd\}.
\end{align}
Let \(\fcc{\cf}\) be the set of feasible cost constraints for \(\pdis{\inpS}\):
\begin{align}
\notag
\fcc{\cf}
&\DEF\{\cc\in \reals{\geq0}^{\lnd}:\exists\mP\in\pdis{\inpS} \mbox{~s.t.~}\sum\nolimits_{\dinp} \mP(\dinp) \cf(\dinp)\leq \cc\}.
\end{align}
Then \(\fcc{\cf}\) is a convex set with non-empty interior.
 A cost function \(\cf\) for a product channel \(\Wm\) is said to be \emph{additive} iff 
there exists a \(\cf_{\tin}:\inpS_{\tin}\to \reals{\geq0}^{\lnd}\) for each \(\tin\in\tinS\) such that
\begin{align}
\notag
\cf(\dinp)
&=\sum\nolimits_{\tin\in\tinS} \cf_{\tin}(\dinp_{\tin})  
&
&\forall \dinp\in\inpS. 
\end{align}
An encoding function \(\enc\), hence the corresponding code, is said to satisfy 
the cost constraint \(\cc\) iff \(\vee_{\dmes\in\mesS} \cf(\enc(\dmes))\leq \cc\).
A code on a product channel \(\Wm:\oprod{{\tin\in\tinS}}\inpS_{\tin}\to \pmea{\outA}\) 
is said to satisfy an empirical distribution constraint \(\cset\subset\pdis{\inpS_{1}}\) 
iff the empirical distribution, i.e. type or composition, of \(\enc(\dmes)\) is in \(\cset\)
for all \(\dmes\in\mesS\).

\subsection{Main Result}
\begin{assumption}\label{assumption:cost-ologn}
\(\{(\Wmn{\tin},\cf_{\tin},\cc_{\tin})\}_{\tin\in\integers{+}}\) is an ordered sequence of 
channels with associated cost functions and cost constraints satisfying the following 
condition: \(\exists \blx_{0}\in \integers{+}, K\in\reals{+}\) s.t.
\begin{align}
\notag
\max\nolimits_{\tin:\tin\leq \blx} \CRC{\frac{1}{2}}{\Wmn{\tin}}{\cc_{\blx}}
&\leq K \ln(\blx)
&
&\mbox{and}
&
&\cc_{\blx}\!\in\!\inte{\fcc{\cf_{[1,\blx]}}}
\end{align}
for all  \(\forall \blx\geq\blx_{0}\)
where \(\cf_{[1,\blx]}(\dinp_{[1,\blx]})=\sum_{\tin=1}^{\blx} \cf_{\tin}(\dinp_{\tin})\).
\end{assumption}
\begin{theorem}\label{thm:CCexponent}
Let \(\{\!(\!\Wmn{\tin},\!\cf_{\tin},\!\cc_{\tin})\!\}_{\tin\in\integers{+}}\!\) be a sequence 
satisfying  Assumption \ref{assumption:cost-ologn},
\(\rno_{0},\rno_{1}\) be orders satisfying \(0\!<\!\rno_{0}\!<\!\rno_{1}\!<\!1\)
and  \(\varepsilon\!\in\!\reals{\geq0}\).
Then for any sequence of codes \(\{(\enc_{\tin},\dec_{\tin})\}_{\tin\in\integers{+}}\) 
on the product channels \(\{\Wmn{[1,\blx]}\}_{\blx\in\integers{+}}\) satisfying 
\begin{align}
\notag
\vee_{\dmes\in\mesS_{\blx}}  \cf_{[1,\blx]}(\enc_{\tin}(\dmes))
&\leq\cc_{\blx}
&
&\forall \blx\in \integers{+}
%\label{eq:thm:CCexponent-hypothesis}
\\
\notag
\CRC{\rno_{0}}{\Wmn{[1,\blx]}}{\cc_{\blx}}\!+\!\varepsilon\!\ln^{2}\!\blx
&\leq\!\ln\!\tfrac{M_{\blx}}{L_{\blx}}
\!\leq\!
\CRC{\rno_{1}}{\Wmn{[1,\blx]}}{\cc_{\blx}}
&
&\forall \blx\geq\blx_{0} 
\end{align}
there exists a \(\tau\in \reals{+}\) and an \(\blx_{1}\geq\blx_{0}\) such that 
\begin{align}
%\label{eq:thm:CCexponent}
\notag
\Pem{av}^{(\blx)}
&\geq \blx^{-\tau}  e^{-\spe{\ln\frac{M_{\blx}}{L_{\blx}},\Wmn{[1,\blx]},\cc_{\blx}}} 
&
&\forall \blx\geq\blx_{1}
\end{align}
where \(\spe{\rate,\Wm,\cc}=\sup\nolimits_{\rno\in (0,1)} \tfrac{1-\rno}{\rno} \left(\CRC{\rno}{\Wm\!}{\cc}-\rate\right)\).
\end{theorem}
Theorem \ref{thm:CCexponent} follows from Lemma \ref{lem:avspherepacking} and Lemma \ref{lem:CCaugustin},
through an analysis similar to the one in \cite[\S \ref*{B-sec:proof:thm:productexponent}]{nakiboglu16B}.
An asymptotic result similar to Theorem \ref{thm:CCexponent} for codes on stationary memoryless channels with 
convex empirical distribution constraints can be proved using Lemma \ref{lem:avspherepacking} and 
the bound given in equation (\ref{eq:lem:Saugustin}).
\section{The Augustin Information and Capacity}\label{sec:center}
\vspace{-1mm}
\(\forall\rno\in\reals{+},\mW,\mQ\in\fmea{\outA}\), \emph{the order \(\rno\) 
Renyi divergence} is 
\begin{align}
\notag
\RD{\rno}{\mW}{\mQ}
&\DEF \begin{cases}
\tfrac{1}{\rno-1}\ln \int (\der{\mW}{\rfm})^{\rno} (\der{\mQ}{\rfm})^{1-\rno} \rfm(\dif{\dout})  
&\rno\neq 1\\
\int  \der{\mW}{\rfm}\left[ \ln\der{\mW}{\rfm} -\ln \der{\mQ}{\rfm}\right] \rfm(\dif{\dout})  
&\rno=1
\end{cases}
\end{align}
where \(\rfm\) is any measure s.t. \(\mW\AC\rfm\),\(\mQ\AC\rfm\).
If \(\RD{\!\rno}{\!\mW}{\!\mQ}\!<\!\infty\) then \emph{the order \(\rno\) tilted probability measure} 
\(\vma{\rno}{\mW,\mQ}\) is 
\begin{align}
\notag
%\label{eq:def:Atiltedprobabilitymeasure}
\der{\vma{\rno}{\mW,\mQ}}{\rfm}
&\DEF e^{(1-\rno)\RD{\rno}{\mW}{\mQ}}
(\der{\mW}{\rfm})^{\rno} (\der{\mQ}{\rfm})^{1-\rno}.
\end{align}
\vspace{-5mm}
\subsection{The Augustin Information and Mean}\label{sec:information}
\vspace{-1mm}
\begin{definition}\label{def:information}
For any \(\rno\in\reals{+}\), \(\!\Wm:\inpS\to \pmea{\outA}\), and \(\mP\in \pdis{\inpS}\)
\emph{the order \(\rno\) Augustin information for the prior \(\mP\)} is
\begin{align}
%\label{eq:def:information}
\notag
\RMI{\rno}{\mP}{\Wm}
&\DEF \inf\nolimits_{\mQ\in \pmea{\outA}} \CRD{\rno}{\Wm}{\mQ}{\mP}
\end{align}
where \(\CRD{\rno}{\Wm}{\mQ}{\mP}\DEF \sum\nolimits_{\dinp\in \inpS}  \mP(\dinp) \RD{\rno}{\Wm(\dinp)}{\mQ}\).
\end{definition}

Whenever it exists, the uniqueness of \(\qmn{\rno,\mP}\in\pmea{\outA}\) satisfying 
\(\RMI{\rno}{\mP}{\Wm}=\CRD{\rno}{\Wm}{\qmn{\rno,\mP}}{\mP}\) 
follows from the strict convexity of \(\RD{\rno}{\mW}{\mQ}\) in \(\mQ\),
i.e. \cite[Thm 12]{ervenH14}.
Such a \(\qmn{\rno,\mP}\) is called \emph{the order \(\rno\) Augustin mean for the prior \(\mP\)}.
If \(\abs{\outS}<\infty\) then \(\pmea{\outA}\) is  compact and the existence of 
\(\qmn{\rno,\mP}\) follows from the lower semicontinuity of 
\(\RD{\rno}{\mW}{\mQ}\) in \(\mQ\), i.e \cite[Lem \ref*{A-lem:divergencelsc}]{nakiboglu16A}, and 
the extreme value theorem \cite[Ch3\S12.2]{kolmogorovfomin75}.

Lemma \ref{lem:information} asserts the existence of a unique \(\qmn{\rno,\mP}\) for arbitrary channels 
and describes \(\qmn{\rno,\mP}\) via the identities it has to 
satisfy. 
Part (\ref{information:orderone}) is well known;
part (\ref{information:orderzeroone}) is due to\footnote{\cite[34.2]{augustin78} 
claims eq. (\ref{eq:lem:information:iteration}) for \(\qgn{1,\mP}\) instead of \(\qgn{\rno,\mP}\). 
We could not confirm the correctness of Augustin's proof of \cite[34.2]{augustin78}, see \cite{nakiboglu16C}.} 
Augustin \cite[34.2]{augustin78}.
A generalization of Lemma \ref{lem:information} for all \(\rno\!\in\!\reals{+}\) is proved in \cite{nakiboglu16C}.
\begin{definition}
For any \(\rno\in\reals{+}\), \(\Wm\!:\!\inpS\!\to\!\pmea{\outA}\), and \(\mP\!\in\!\pdis{\inpS}\),
\begin{itemize}
\item \(\Aop{\rno}{\mP}{\cdot}:\{\mQ\in \fmea{\outA}:\CRD{\rno}{\Wm}{\mQ}{\mP}<\infty\}\to\pmea{\outA}\) is 
\begin{align}
\notag
\Aop{\rno}{\mP}{\mQ}
&\DEF\sum\nolimits_{\dinp}\mP(\dinp) \vma{\rno}{\Wm(\dinp),\mQ}.
\end{align}
Furthermore, 
\(\Aopi{\rno}{\mP}{\ind+1}{\mQ}\DEF\Aop{\rno}{\mP}{\Aopi{\rno}{\mP}{\ind}{\mQ}}\) 
for \(\ind\in\integers{+}\).

\item \(\mmn{\rno,\mP}\in \fmea{\outA}\) and \(\qgn{\rno,\mP}\in\pmea{\outA}\) are given by
\begin{align}
\notag
\der{\mmn{\rno,\mP}}{\rfm}
&\DEF \left[\sum\nolimits_{\dinp} \mP(\dinp) \left(\der{\Wm{(\dinp)}}{\rfm}\right)^{\rno}  \right]^{\frac{1}{\rno}}
&
&
&
\qgn{\rno,\mP}
&\DEF\tfrac{\mmn{\rno,\mP}}{\lon{\mmn{\rno,\mP}}}
\end{align}
where \(\rfm\) is any measure for which \((\sum_{\dinp}\mP(\dinp)\Wm(\dinp))\AC\rfm\).
\end{itemize}
\end{definition}

\begin{lemma}\label{lem:information}
For any \(\Wm:\inpS\to \pmea{\outA}\) and \(\mP\in\pdis{\inpS}\),
\begin{enumerate}[(a)]
\item\label{information:orderone}
\(\RMI{1}{\mP}{\Wm}=\CRD{1}{\Wm}{\qmn{1,\mP}}{\mP}\) for \(\qmn{1,\mP}\DEF\sum\nolimits_{\dinp} \mP(\dinp) \Wm(\dinp)\).
\begin{align}
\label{eq:lem:informationEHB-one}
\CRD{1}{\!\Wm}{\mQ}{\mP}-\RMI{1}{\mP}{\!\Wm}
&\!=\!\RD{1}{\qmn{1,\mP}}{\mQ}
&
&\forall\!\mQ\!\in\pmea{\outA}.
\end{align}
\item\label{information:orderzeroone}
\(\forall\!\rno\!\!\in\!\!(0,\!1)\exists!\qmn{\rno,\mP}\) s.t. \(\RMI{\rno}{\mP}{\!\Wm}\!=\!\CRD{\rno}{\!\Wm\!}{\!\qmn{\rno,\mP}\!}{\mP}\).
\(\qmn{\rno,\mP}\!\sim\!\qmn{1,\mP}\),
\begin{align}
\label{eq:lem:informationEHB}
\CRD{\rno}{\Wm}{\mQ}{\mP}\!-\!\RMI{\rno}{\mP}{\Wm}
&\!\geq\!\RD{\rno}{\qmn{\rno,\mP}}{\mQ}
&
\forall\!\mQ\!\!\in\!\pmea{\outA}&
\\
\label{eq:lem:information:fixedpoint}
\Aop{\rno}{\mP}{\qmn{\rno,\mP}}
&\!=\qmn{\rno,\mP}
\\
\label{eq:lem:information:iteration}
\lim\limits_{\jnd\to\infty}\!
\lon{\qmn{\rno,\mP}\!-\!\Aopi{\rno}{\mP}{\jnd}{\qgn{\rno,\mP}}}
\!&\!=\!0.
\end{align}
Furthermore, if a \(\mQ\in\pmea{\outA}\) satisfying \(\qmn{1,\mP}\AC \mQ\) 
is a fixed point of \(\Aop{\rno}{\mP}{\cdot}\) then \(\mQ=\qmn{\rno,\mP}\).
\item\label{information:product}
If \(\rno\in(0,1]\), \(\Wm\) is a product channel for a finite index set \(\tinS\),
and \(\mP\) is of the form  \(\oprod{\tin\in\tinS} \pmn{\tin}\) 
for \(\pmn{\tin}\in \pdis{\inpS_{\tin}}\) then 
\begin{align}
\label{eq:information:product}
\qmn{\rno,\mP}
&\!=\!\oprod{\tin\in\tinS}\!\qmn{\rno,\pmn{\tin}}
&
&
&
\RMI{\rno}{\mP}{\Wm}
&\!=\!\sum\nolimits_{\tin\in \tinS}\!\RMI{\rno}{\pmn{\tin}}{\Wmn{\tin}}.
\end{align} 
\end{enumerate}
\end{lemma}

\subsection{The Constrained Augustin Capacity and Center}\label{sec:capacity}
\begin{definition}\label{def:capacity}
For any \(\rno\!\in\!\reals{+}\), \(\Wm:\inpS\!\to\!\pmea{\outA}\), and 
\(\cset\!\subset\!\pdis{\inpS}\),
\emph{the order \(\rno\) Augustin capacity of \(\Wm\) for constraint set \(\cset\)} is 
\begin{align}
%\label{eq:def:capacity}
\notag
\CRC{\rno}{\Wm\!}{\cset}
&\DEF \sup\nolimits_{\mP \in \cset}  \RMI{\rno}{\mP}{\Wm}.
\end{align}
\end{definition}
Using the definition of \(\RMI{\rno}{\mP}{\Wm}\) we get
\begin{align}
%\label{eq:capacity}
\notag
\CRC{\rno}{\Wm\!}{\cset}
&=\sup\nolimits_{\mP \in \cset}\inf\nolimits_{\mQ\in\pmea{\outA}} \CRD{\rno}{\Wm}{\mQ}{\mP}.
\end{align}

Proofs of the propositions presented in this subsection can be found in \cite{nakiboglu16C}.
They are very similar to the proofs of the corresponding claims in 
\cite[\S\ref*{A-sec:capacity}, \S\ref*{A-sec:center}, \S\ref*{A-sec:constrainedcapacity}]{nakiboglu16A}
for Renyi capacity;
we invoke Lemma \ref{lem:information} instead of \cite[Lem \ref*{A-lem:information:def}]{nakiboglu16A}.

\begin{lemma}\label{lem:capacityO}.
For any \(\Wm:\inpS\to \pmea{\outA}\) and \(\cset\subset\pdis{\inpS}\)
\begin{enumerate}[(a)]
\item\label{capacityO-increasing}
\(\CRC{\rno}{\Wm\!}{\cset}:(0,1]\!\to\![0,\infty]\) 
is increasing and continuous
\item\label{capacityO-decreasing}
\(\tfrac{1-\rno}{\rno}\!\CRC{\rno}{\Wm\!}{\cset}:(0,1)\!\to\![0,\infty]\) 
is decreasing and continuous
\item\label{capacityO-finiteness}
\(\exists\rno\!\in\!(0,1)\) ~s.t.~\(\CRC{\rno}{\Wm\!}{\cset}\!<\!\infty\) iff
\(\CRC{\rnf}{\Wm}{\cset}\!<\!\infty~\forall\rnf\!\in\!(0,1)\).
\end{enumerate}
\end{lemma}

\begin{theorem}\label{thm:minimax}
\(\forall\rno\!\in\!(0,1]\),\(\Wm\!:\!\inpS\!\to\!\pmea{\outA}\), and
convex \(\cset\!\subset\!\pdis{\inpS}\), 
\begin{align}
%\label{eq:thm:minimax}
\notag
\sup\limits_{\mP \in \cset} \inf\limits_{\mQ \in \pmea{\outA}} 
 \CRD{\rno}{\Wm}{\mQ}{\mP}
&=
\inf\limits_{\mQ \in \pmea{\outA}} \sup\limits_{\mP \in \cset} 
 \CRD{\rno}{\Wm}{\mQ}{\mP}.
\end{align}
If \(\CRC{\rno}{\Wm\!}{\cset}<\infty\) then 
\(\exists!\qmn{\rno,\Wm,\cset}\in\pmea{\outA}\),
called the order \(\rno\) Augustin center of \(\Wm\) for the constraint set \(\cset\), 
such that
\begin{align}
%\label{eq:thm:minimaxcenter}
\notag
\CRC{\rno}{\Wm\!}{\cset}
&=\sup\nolimits_{\mP \in \cset} \CRD{\rno}{\Wm}{\qmn{\rno,\Wm,\cset}}{\mP}.
\end{align}
If \(\lim_{\ind\!\to\!\infty}\!\RMI{\rno}{\pma{}{(\ind)}}{\Wm}\!=\!\CRC{\rno}{\!\Wm\!}{\cset}<\infty\)
for a \(\{\!\pma{}{(\ind)}\!\}_{\ind\in\integers{+}}\!\subset\!\cset\) then
\(\{\qmn{\rno,\pma{}{(\ind)}}\}_{\ind\in\integers{+}}\)  
is a Cauchy sequence  for the total variation metric on \(\pmea{\outA}\) 
and \(\qmn{\rno,\Wm\!,\cset}\) is its unique limit point.
\end{theorem}
Lemma \ref{lem:information} and  Theorem \ref{thm:minimax} 
imply for all \(\rno\!\in\!(0,1]\), \(\mP\!\in\!\cset\) that
\begin{align}
%\label{eq:capacityLB}
\notag
\CRC{\rno}{\Wm\!}{\cset}-\RMI{\rno}{\mP}{\Wm}
&\geq  \RD{\rno}{\qmn{\rno,\mP}}{\qmn{\rno,\Wm,\cset}}.
\end{align}
Using Lemma \ref{lem:information} and Theorem \ref{thm:minimax} we can prove the following 
Erven-Harremoes bound for Augustin capacity.
\begin{lemma}\label{lem:EHB}
For any \(\rno\!\in\!(0,1]\),\(\Wm\!:\!\inpS\!\to\!\pmea{\outA}\), and
convex \(\cset\!\subset\!\pdis{\inpS}\) s.t. 
\(\CRC{\rno}{\Wm\!}{\cset}<\infty\), and \(\mQ \in \pmea{\outA}\)
\begin{align}
%\label{eq:lem:EHB}
\notag
\sup\nolimits_{\mP \in \cset} \CRD{\rno}{\Wm}{\mQ}{\mP}
&\geq  \CRC{\rno}{\Wm\!}{\cset}+\RD{\rno}{\qmn{\rno,\Wm,\cset}}{\mQ}.
\end{align}
\end{lemma}

Erven-Harremoes bound, the continuity of \(\CRC{\rno}{\Wm\!}{\cset}\) in \(\rno\), and  
Pinsker's inequality imply the continuity of \(\qmn{\rno,\Wm,\cset}\) in  \(\rno\) for the 
total variation topology  on \(\pmea{\outA}\).

\begin{lemma}\label{lem:centercontinuity} 
For any \(\rnt\!\in\!(0,1]\),\(\Wm\!:\!\inpS\!\to\!\pmea{\outA}\),
convex \(\cset\!\subset\!\pdis{\inpS}\) 
s.t. \(\CRC{\rnt}{\Wm}{\cset}<\infty\),
and  \(\rno\), \(\rnf\) satisfying \(0<\rno<\rnf\leq\rnt\),
\begin{align}
%\label{eq:lem:centercontinuity}
\notag
\CRC{\rnf}{\Wm}{\cset}-\CRC{\rno}{\Wm\!}{\cset}
&\geq 
\RD{\rno}{\qmn{\rno,\Wm,\cset}}{\qmn{\rnf,\Wm,\cset}}.
\end{align}
Furthermore, \(\qmn{\rno,\Wm,\cset}:(0,\rnt]\to \pmea{\outA}\) 
is continuous in \(\rno\) for the total variation topology  on \(\pmea{\outA}\).
\end{lemma}

\begin{lemma}\label{lem:capacityproduct}
For any \(\rno\in(0,1]\), product channel \(\Wm\) for a finite index set \(\tinS\),
convex sets \(\cset_{\tin}\subset\pdis{\inpS_{\tin}}\) for each \(\tin\in\tinS\), and
\(\cset=\conv{\{\oprod{\tin\in\tinS}\pmn{\tin}:\pmn{\tin}\in\pdis{\inpS_{\tin}}~\forall\tin\in\tinS\}}\)
\begin{align}
%\label{eq:lem:capacityproduct}
\notag
\CRC{\rno}{\Wm\!}{\cset}
&=\sum\nolimits_{\tin\in\tinS}\CRC{\rno}{\Wmn{\tin}}{\cset_{\tin}}.
\end{align}
Furthermore, if \(\CRC{\rno}{\Wm\!}{\cset}<\infty\) then 
\(\qmn{\rno,\Wm,\cset}=\oprod{\tin\in\tinS}\qmn{\rno,\Wmn{\tin},\cset_{\tin}}\).
\end{lemma}
\section{The Cost Constrained Augustin Capacity}\label{sec:costconstrained}
With a slight abuse of notation we define the cost constrained Augustin capacity as 
\begin{align}
%\label{eq:def:costcapacity}
\notag
\CRC{\rno}{\Wm\!}{\cc}
&\DEF \sup\nolimits_{\mP\in\cset(\cc)} \RMI{\rno}{\mP}{\Wm}
&
&\forall\cc\in \fcc{\cf}
\end{align}
where \(\cset(\cc)\DEF \{\mP\in\pdis{\inpS}:\sum\nolimits_{\dinp} \mP(\dinp) \cf(\dinp)\leq \cc\}\).
Note that  Theorem \ref{thm:minimax} and Lemmas \ref{lem:EHB} and \ref{lem:centercontinuity}
hold for \(\CRC{\rno}{\Wm\!}{\cc}\) because \(\cset(\cc)\) is a convex set.
We denote Augustin center by \(\qmn{\rno,\Wm\!,\cc}\).

\begin{lemma}\label{lem:CCcapacity}
For any  \(\rno\in(0,1]\), \(\Wm:\inpS\to \pmea{\outA}\), \(\cf:\inpS\to \reals{\geq0}^{\lnd}\),
\begin{enumerate}[(a)]
\item\label{CCcapacity:function} 
\(\CRC{\rno}{\!\Wm\!}{\cc}\!:\!\fcc{\cf}\!\to\![0,\infty]\) is increasing and concave in \(\cc\).
It is either infinite \(\forall\!\cc\!\in\!\inte{\!\fcc{\!\cf}}\)
or finite and continuous\! on \(\!\inte{\!\fcc{\!\cf}}\). 

\item\label{CCcapacity:interior}
If \(\CRC{\rno}{\!\Wm\!}{\cc}<\infty\) for a \(\cc\in \inte{\fcc{\cf}}\) then 
\(\exists\lm_{\rno,\Wm\!,\cc}\in \reals{\geq0}^{\lnd}\) s.t.
\begin{align}
%\label{eq:CCcapacity:interior}
\notag
\CRC{\rno}{\Wm\!}{\cc}+\lm_{\rno,\Wm\!,\cc}\cdot(\tilde{\cc}-\cc)
&\geq \CRC{\rno}{\Wm}{\tilde{\cc}}
&
&\forall \tilde{\cc}\in\fcc{\cf}.
\end{align}
\!\!The set of all such \(\lm_{\rno,\!\Wm\!,\cc}\!\)'s for an \(\rno\) is convex and compact.
\end{enumerate}
\end{lemma}

\begin{lemma}\label{lem:CCcapacityproduct}
For any \(\rno\in(0,1]\), product channel \(\Wm\) for a finite index set \(\tinS\), additive 
cost function \(\cf:\inpS\to\reals{\geq0}^{\lnd}\) satisfying 
\(\cf(\dinp)=\sum_{\tin\in\tinS} \cf_{\tin}(\dinp_{\tin})\) for some 
\(\cf_{\tin}:\inpS_{\tin}\to\reals{\geq0}^{\lnd}\)
and \(\cc\in \fcc{\cf}\)
\begin{align}
\notag
\CRC{\rno}{\Wm\!}{\cc}
&=\sup\left\{\sum\nolimits_{\tin\in \tinS} \CRC{\rno}{\Wmn{\tin}}{\cc_{\tin}}:
\sum\nolimits_{\tin\in \tinS} \cc_{\tin}\leq \cc,~
\cc_{\tin}\in\fcc{\cf_{\tin}} \right\}
\end{align}
If \(\exists\{\cc_{\tin}\}_{\tin\in \tinS}\) s.t.
\(\CRC{\rno}{\Wm\!}{\cc}=\sum\nolimits_{\tin\in \tinS} \CRC{\rno}{\Wmn{\tin}}{\cc_{\tin}}\)
and \(\CRC{\rno}{\Wm\!}{\cc}<\infty\)
then \(\qmn{\rno,\Wm\!,\cc}=\oprod{\tin\in \tinS} \qmn{\rno,\Wmn{\tin},\cc_{\tin}}\).
\end{lemma}

Since Augustin capacity is concave in the cost constraint by Lemma \ref{lem:CCcapacity}-(\ref{CCcapacity:function}),
\(\CRC{\rno}{\Wm\!}{\cc}=\sum\nolimits_{\tin\in \tinS} \CRC{\rno}{\Wmn{\tin}}{\frac{\cc}{\abs{\tinS}}}\)
whenever \(\Wm\) is stationary and \(\cf_{\tin}=\cf_{1}\) for all \(\tin\in\tinS\).
Alternatively, if \(\fcc{\cf_{\tin}}\)'s are closed and 
\(\CRC{\rno}{\Wmn{\tin}}{\cc}\)'s are upper semicontinuous functions of \(\cc\)
on \(\fcc{\cf_{\tin}}\)'s then we can use the extreme value theorem for the upper 
semicontinuous functions to establish the 
existence of a \(\{\cc_{\tin}\}_{\tin\in \tinS}\) s.t.
\(\CRC{\rno}{\Wm\!}{\cc}=\sum\nolimits_{\tin\in \tinS} \CRC{\rno}{\Wmn{\tin}}{\cc_{\tin}}\).
However, such an existence assertion does not hold in general.
\subsection{The A-L Information, Capacity, Center, and Radius}\label{sec:augustinlegendre}
This subsection is a generalization of parts of \cite[Ch. 8]{csiszarkorner}, 
which is confined to \(\abs{\inpS}\!\vee\!\abs{\outS}\!<\!\infty \), \(\rno\!=\!1\), and 
\(\lnd\!=\!1\) case.

For any \(\rno\!\in\!\reals{+}\), \(\Wm\!:\!\inpS\!\to\!\pmea{\outA}\), 
cost function \(\cf\!:\!\inpS\!\to\!\reals{\geq0}^{\lnd}\),
\(\lm\!\in\!\reals{\geq0}^{\lnd}\), and \(\mP\!\in\!\pdis{\inpS}\)
\emph{the order \(\rno\) Augustin-Legendre (A-L) information for prior \(\mP\) and 
Lagrange multiplier \(\lm\)} is
\begin{align}
%\label{eq:def:Linformation}
\notag
\RMIL{\rno}{\mP}{\Wm}{\lm}
&\DEF \RMI{\rno}{\mP}{\Wm}-\lm\cdot \left(\sum\nolimits_{\dinp} \mP(\dinp) \cf(\dinp)\right).
\end{align}
We call \(\RMIL{\rno}{\mP}{\Wm}{\lm}\) A-L information because of the convex
conjugate pair \(\fX_{\rno,\mP}:\reals{\geq0}^{\lnd}\to (-\infty,\infty]\)
and  \(\fX_{\rno,\mP}^{*}:\reals{\leq0}^{\lnd}\to \reals{}\): 
\begin{align}
\notag
\fX_{\rno,\mP}(\cc)
&\DEF
\begin{cases}
-\RMI{\rno}{\mP}{\Wm}		&\cc\geq \EXS{\mP}{\cf}
\\
\infty						&\mbox{else}
\end{cases}
&
&=\sup_{\xi\leq 0} \xi \cdot \cc - \fX_{\rno,\mP}^{*}(\xi) 
\\
\notag
\fX_{\rno,\mP}^{*}(\xi)
&\DEF \sup\nolimits_{\cc \geq 0} \xi \cdot \cc -\fX_{\rno,\mP}(\cc) 
&
&=\xi \cdot \EXS{\mP}{\cf} +\RMI{\rno}{\mP}{\Wm}
\end{align}
Thus one can write \(\CRC{\rno}{\Wm\!}{\cc}\) in terms of \(\RMIL{\rno}{\mP}{\Wm}{\lm}\) as
\begin{align}
%\label{eq:Lcostcapacity}
\notag
\CRC{\rno}{\Wm\!}{\cc}
&=\sup\nolimits_{\mP\in\pdis{\inpS}}\inf\nolimits_{\lm\geq 0} \RMIL{\rno}{\mP}{\Wm}{\lm}+\lm\cdot\cc.
\end{align}

\(\RMIL{\rno}{\mP}{\Wm}{\lm}\) is convex, decreasing and continuous in \(\lm\).
Furthermore, by Lemma \ref{lem:information} for \(\rno\in(0,1]\) we have:
\begin{align}
\notag
\RMIL{\rno}{\mP}{\Wm}{\lm}
&=\CRD{\rno}{\Wm}{\qmn{\rno,\mP}}{\mP}-\lm\cdot \EXS{\mP}{\cf}
\\
\notag
\CRD{\rno}{\Wm}{\mQ}{\mP}-\lm\cdot \EXS{\mP}{\cf}
&\geq \RMIL{\rno}{\mP}{\Wm}{\lm}+\RD{\rno}{\qmn{\rno,\mP}}{\mQ}.
\end{align} 
For any \(\rno\!\in\!(0,1]\), \(\Wm\!:\!\inpS\!\to\!\pmea{\outA}\), 
\(\cf\!:\!\inpS\!\to\!\reals{\geq0}^{\lnd}\),
and \(\lm\!\in\!\reals{\geq0}^{\lnd}\),
\emph{the A-L capacity} \(\RCL{\rno}{\Wm}{\lm}\)
and 
\emph{the A-L radius} \(\RRL{\rno}{\Wm}{\lm}\)
are given by
\begin{align}
%\label{eq:def:Lcapacity}
\notag
\RCL{\rno}{\Wm}{\lm}
&\DEF \sup\nolimits_{\mP\in \pdis{\inpS}} \RMIL{\rno}{\mP}{\Wm}{\lm}
\\
%\label{eq:def:Lradius}
\notag
\RRL{\rno}{\Wm}{\lm}
&\DEF \inf\nolimits_{\mQ\in\pmea{\outA}} \sup\nolimits_{\dinp\in \inpS} \RD{\rno}{\Wm(\dinp)}{\mQ}-\lm\cdot\cf(\dinp). 
\end{align}
Using the definition of \(\RMIL{\rno}{\mP}{\Wm}{\lm}\), \(\RMI{\rno}{\mP}{\Wm}\) and \(\RRL{\rno}{\Wm}{\lm}\) we get  
\begin{align}
%\label{eq:Lcapacity}
\notag
\RCL{\rno}{\Wm}{\lm}
&\!=\!\sup\nolimits_{\mP \in \pdis{\inpS}}\inf\nolimits_{\mQ\in\pmea{\outA}} \CRD{\rno}{\Wm}{\mQ}{\mP}\!-\!\lm\cdot \EXS{\mP}{\cf}
\\
%\label{eq:Lradius}
\notag
\RRL{\rno}{\Wm}{\lm}
&\!=\!\inf\nolimits_{\mQ\in\pmea{\outA}}\sup\nolimits_{\mP \in \pdis{\inpS}} \CRD{\rno}{\Wm}{\mQ}{\mP}\!-\!\lm\cdot \EXS{\mP}{\cf}.
\end{align}

\begin{lemma}\label{lem:Lcapacity}
For any \(\rno\!\in\!(0,1]\), \(\Wm\!:\!\inpS\!\to\!\pmea{\outA}\), 
\(\cf\!:\!\inpS\!\to\!\reals{\geq0}^{\lnd}\),
\begin{enumerate}[(a)]
\item\label{Lcapacity:function}
\(\RCL{\rno}{\!\Wm}{\lm}\) is convex, decreasing and lower semicontinuous in \(\lm\)
on \(\reals{\geq0}^{\lnd}\) and continuous in \(\lm\) on 
\(\{\!\lm\!:\!\exists\epsilon\!>\!0\!~s.t.\!~ \RCL{\rno}{\!\Wm}{\lm-\epsilon\uc}\!<\!\infty\!\}\).
\item\label{Lcapacity:maxmin}  
\(\CRC{\rno}{\Wm\!}{\cc} \leq \inf\nolimits_{\lm\geq 0} \RCL{\rno}{\!\Wm}{\lm}+\lm\cdot \cc\) 
for all \(\cc\in \fcc{\cf}\).
\item\label{Lcapacity:minimax} 
\(\CRC{\!\rno}{\!\Wm\!}{\cc}\!=\!\inf\nolimits_{\lm\geq0}\!\RCL{\rno}{\Wm}{\lm}\!+\!\lm\!\cdot\!\cc\)
if either \(\abs{\inpS}\!<\!\infty\) or \(\cc\in\inte{\fcc{\cf}}\).
\item\label{Lcapacity:interior} 
If \(\exists \cc\!\in\!\inte{\fcc{\cf}}\) s.t. \(\CRC{\rno}{\Wm\!}{\cc}\!<\!\infty\) 
then \(\forall\cc\!\in\!\inte{\fcc{\cf}}~\exists\lm\in\!\reals{\geq0}^{\lnd}\) 
s.t.  \(\CRC{\rno}{\Wm\!}{\cc}=\RCL{\rno}{\Wm}{\lm}+\lm\cdot \cc\).
\item\label{Lcapacity:optimal}
If \(\CRC{\rno}{\Wm\!}{\cc}\!=\!\RCL{\rno}{\Wm}{\lm}\!+\!\lm\cdot\cc<\infty\) 
for  a \((\cc,\lm)\in \fcc{\cf}\times \reals{\geq0}^{\lnd}\), and 
\(\lim\nolimits_{\ind \to \infty} \RMI{\rno}{\pma{}{(\ind)}}{\Wm}\!=\!\CRC{\rno}{\Wm\!}{\cc}\)
for a \(\{\pma{}{(\ind)}\}_{\ind\in\integers{+}}\!\subset\!\cset(\cc)\) then
\(\lim\nolimits_{\ind \to \infty} \RMIL{\rno}{\pma{}{(\ind)}}{\Wm}{\lm}\!=\!\RCL{\rno}{\Wm}{\lm}\).
\end{enumerate}
\end{lemma}
If \(\exists\lm\in\reals{\geq0}\) s.t. \(\RCL{\rno}{\Wm}{\lm}<\infty\) then 
\(\CRC{\rno}{\Wm\!}{\cc}<\infty\) \(\forall\cc\in\fcc{\cf}\) 
by Lemma \ref{lem:Lcapacity}-(\ref{Lcapacity:function}).
However, the converse claim is not true. 
There are cases for which \(\CRC{\rno}{\Wm\!}{\cc}\) is finite for all \(\cc\in\fcc{\cf}\), yet 
\(\RCL{\rno}{\Wm}{\lm}\) is infinite for \(\lm\) small 
enough.\footnote{In \cite[\!\S33\!-\!\S35]{augustin78}, Augustin considers bounded \(\cf\)'s
of the form \(\cf\!:\!\inpS\!\to\![0,1]^{\lnd}\). In that case, it is easy to see that 
if \(\exists\cc\in\inte{\fcc{\cf}}\) s.t. \(\CRC{\rno}{\Wm\!}{\cc}<\infty\) then
\(\sup_{\cc\in \fcc{\cf}}\CRC{\rno}{\Wm\!}{\cc}=\CRC{\rno}{\Wm}{\uc}<\infty\) and
\(\RCL{\rno}{\Wm}{\lm}<\infty\) for all \(\lm\in\reals{\geq0}^{\lnd}\).}
The equality given in (\ref{Lcapacity:minimax}) might not hold 
if \(\cc\in \fcc{\cf}\setminus\inte{\fcc{\cf}}\) and \(\abs{\inpS}=\infty\).

\begin{theorem}\label{thm:Lminimax}
\(\forall\rno\!\in\!(0,1]\), \(\Wm\!:\!\inpS\!\to\!\pmea{\outA}\), \(\cf\!:\!\inpS\!\to\!\reals{\geq0}^{\lnd}\),
\(\lm\!\in\!\reals{\geq0}^{\lnd}\),
\begin{align}
%\label{eq:thm:Lminimax}
\notag
\RCL{\rno}{\Wm}{\lm}
&=\RRL{\rno}{\Wm}{\lm}.
\end{align}
If \(\RCL{\rno}{\Wm}{\lm}<\infty\)  then \(\exists!\qma{\rno,\Wm}{\lm}\in\pmea{\outA}\),
called the order \(\rno\) A-L center of \(\Wm\) for the Lagrange multiplier \(\lm\),
such that
\begin{align}
%\label{eq:thm:Lminimaxcenter}
\notag
\RCL{\rno}{\Wm}{\lm}
%&=\sup\nolimits_{\mP \in \pdis{\inpS}} 
%\CRD{\rno}{\Wm}{\qma{\rno,\Wm}{\lm}}{\mP}-\lm\cdot\EXS{\mP}{\cf}
%\\
%\label{eq:thm:Lminimaxradiuscenter}
&=\sup\nolimits_{\dinp \in \inpS} \RD{\rno}{\Wm(\dinp)}{\qma{\rno,\Wm}{\lm}}-\lm\cdot\cf(\dinp).
\end{align}
If \(\lim_{\ind \to \infty}\!\RMIL{\rno}{\pma{}{(\ind)}\!}{\!\Wm}{\lm}\!=\!\RCL{\!\rno}{\!\Wm}{\!\lm}<\infty\)
for a \(\{\!\pma{}{(\ind)}\!\}_{\!\ind\in\integers{+}}\!\subset\!\pdis{\!\inpS\!}\)  then  corresponding 
\(\{\qmn{\rno,\pma{}{(\ind)}}\}_{\ind\in\integers{+}}\)is a Cauchy 
sequence  for the total variation metric on \(\pmea{\outA}\) and \(\qma{\rno,\Wm}{\lm}\) is its unique 
limit point.
\end{theorem}
\begin{lemma}\label{lem:Lcenter} 
If \(\rno\!\in\!(0,1]\), \(\Wm\!:\!\inpS\!\to\!\pmea{\outA}\), \(\cf\!:\!\inpS\!\to\!\reals{\geq0}^{\lnd}\),
\(\cc\!\in\!\fcc{\cf}\) s.t. \(\CRC{\rno}{\Wm\!}{\cc}<\infty\) and \(\lm\!\in\!\reals{\geq0}^{\lnd}\) s.t.
\(\CRC{\rno}{\Wm\!}{\cc}=\RCL{\rno}{\Wm}{\lm}+\lm\cdot\cc\) then \(\qmn{\rno,\Wm\!,\cc}=\qma{\rno,\Wm}{\lm}\).
\end{lemma}
\begin{lemma}\label{lem:Lcapacityproduct}
\(\forall \rno\!\in\!(0,\!1]\), product channel \(\Wm\) for finite index set \(\tinS\),
and \(\cf\) satisfying 
\(\cf(\dinp)\!=\!\sum_{\tin\in\tinS}\!\cf_{\tin}\!(\dinp_{\tin})\) for some 
\(\cf_{\tin}\!:\!\inpS_{\tin}\!\to\!\reals{\geq0}^{\lnd}\), 
\begin{align}
%\label{eq:lem:Lcapacityproduct}
\notag
\RCL{\rno}{\Wm}{\lm}
&=\sum\nolimits_{\tin\in \tinS} \RCL{\rno}{\Wmn{\tin}}{\lm}
&
&\forall \lm \in \reals{\geq0}^{\lnd}. 
\end{align}
If \(\RCL{\rno}{\Wm}{\lm}<\infty\) then \(\qma{\rno,\Wm}{\lm}=\oprod{\tin\in \tinS} \qma{\rno,\Wmn{\tin}}{\lm}\).
\end{lemma}
Recall that the product structure assertion for \(\qmn{\rno,\!\Wm\!,\cc}\) in Lemma \ref{lem:CCcapacityproduct},
was qualified by the existence of a \(\{\cc_{\tin}\}_{\tin\in \tinS}\)  satisfying
\(\sum\nolimits_{\tin\in \tinS} \CRC{\rno}{\Wmn{\tin}}{\cc_{\tin}}=\CRC{\rno}{\Wm\!}{\cc}<\infty\).
In Lemma \ref{lem:Lcapacityproduct}, on the other hand, the product structure assertion for \(\qma{\rno,\Wm}{\lm}\) 
is qualified only by \(\RCL{\rno}{\Wm}{\lm}<\infty\). 

\subsection{The R-G Information, Mean, Capacity, and Center}\label{sec:renyigallager}
For any \(\!\rno\!\in\!\reals{+}\!\!\setminus\!\{\!1\!\}\), \(\Wm\!:\!\inpS\!\to\!\pmea{\outA}\), 
cost function \(\cf\!:\!\inpS\!\to\!\reals{\geq0}^{\lnd}\),
\(\lm\!\in\!\reals{\geq0}^{\lnd}\), and \(\mP\!\in\!\pdis{\inpS}\)
\emph{the order \(\rno\) Renyi-Gallager (R-G) information for prior \(\mP\) and 
Lagrange multiplier  \(\lm\)} is
\begin{align}
%\label{eq:def:Ginformation}
\notag
\GMIL{\rno}{\mP}{\Wm}{\lm}
&\DEF \inf\nolimits_{\mQ\in \pmea{\outA}} \RD{\rno}{\mP \circ \Wm e^{\frac{1-\rno}{\rno}\lm\cdot\cf}}{\mP\otimes \mQ}.
\end{align}
\emph{The order \(\rno\) R-G capacity for Lagrange multiplier \(\lm\)} is
\begin{align}
%\label{eq:def:Gcapacity}
\notag
\GCL{\rno}{\Wm}{\lm}
&\DEF \sup\nolimits_{\mP\in \pdis{\inpS}} \GMIL{\rno}{\mP}{\Wm}{\lm}.
\end{align}
Using the definition of \(\GMIL{\rno}{\mP}{\Wm}{\lm}\) and \(\GCL{\rno}{\Wm}{\lm}\) we get
\begin{align}
%\label{eq:Gcapacity}
\notag
\GCL{\rno}{\Wm}{\lm}
&=\sup\nolimits_{\mP \in \pdis{\inpS}}\inf\nolimits_{\mQ\in \pmea{\outA}} 
\RD{\rno}{\mP \circ \Wm e^{\frac{1-\rno}{\rno}\lm\cdot\cf}}{\mP\otimes \mQ}.
\end{align}

Using the concavity of log function and Jensen's inequality one can show that 
\(\RMIL{\rno}{\mP}{\Wm}{\lm}\geq \GMIL{\rno}{\mP}{\Wm}{\lm}\) for
\(\rno\in(0,1)\) and 
\(\RMIL{\rno}{\mP}{\Wm}{\lm}\leq \GMIL{\rno}{\mP}{\Wm}{\lm}\) for
\(\rno\in(1,\infty)\).
On the other hand, one can show by substitution that  
\(\forall\mQ\in\pmea{\outA}\) and \(\rno\!\in\!\reals{+}\!\setminus\{\!1\!\}\),
\begin{align}
\notag
\GMIL{\rno}{\mP}{\Wm}{\lm}
&=\RD{\rno}{\mP \circ \Wm e^{\frac{1-\rno}{\rno}\lm\cdot\cf}}{\mP\otimes \qga{\rno,\Wm}{\lm}}
\\
\notag
\RD{\rno}{\mP \circ \Wm e^{\frac{1-\rno}{\rno}\lm\cdot\cf}}{\mP\otimes \mQ}
&=\GMIL{\rno}{\mP}{\Wm}{\lm}
+\RD{\rno}{\qga{\rno,\mP}{\lm}}{\mQ}
\end{align}
where \(\qga{\rno,\mP}{\lm}\) is \emph{the R-G mean} given in terms of \(\mma{\rno,\mP}{\lm}\) as follows,
\begin{align}
\notag
\qga{\rno,\mP}{\lm}
&\DEF \tfrac{\mma{\rno,\mP}{\lm}}{\lon{\mma{\rno,\mP}{\lm}}}
&
\der{\mma{\rno,\mP}{\lm}}{\rfm}
&\DEF\left[\!\sum\nolimits_{\dinp}\!\!\mP(\dinp)e^{(1-\rno)\lm \cdot \cf(\dinp)}\!\!\left(\!\der{\Wm(\dinp)}{\rfm}\!\right)^{\!\!\rno}\right]^{\!\frac{1}{\rno}}\!\!.
\end{align}
For \(\lm=0\uc\), R-G information and mean are equal to the corresponding 
Renyi information and mean analyzed in \cite{nakiboglu16A}. 
Following a similar analysis one can show that a minimax theorem similar 
to \cite[Thm \ref*{A-thm:minimax}]{nakiboglu16A} holds for R-G quantities:
\begin{align}
\notag
\GCL{\rno}{\Wm}{\lm}
&=\inf\nolimits_{\mQ\in \pmea{\outA}}\sup\nolimits_{\mP\in\pdis{\inpS}}  \RD{\rno}{\mP \circ \Wm e^{\frac{1-\rno}{\rno}\lm\cdot\cf}}{\mP\otimes \mQ}
\\
\notag
&=\inf\nolimits_{\mQ\in\pmea{\outA}}\!\sup\nolimits_{\dinp\in \inpS}\!\RD{\rno}{\Wm(\dinp)}{\mQ}\!-\!\lm\cdot\cf(\dinp).
\end{align}
Then \(\GCL{\rno}{\Wm}{\lm}=\RCL{\rno}{\Wm}{\lm}\)~ \(\forall\lm\in\reals{\geq0}^{\lnd},\rno\in(0,1)\)
by Theorem \ref{thm:Lminimax}.
\section{Sphere Packing Bounds}\label{sec:outerbound}
\begin{lemma}\label{lem:HT}
For any 
\(\mW\!=\!\wmn{1}\!\otimes\!\cdots\!\otimes\!\wmn{\blx}\),
\(\mQ\!=\!\qmn{1}\!\otimes\!\cdots\!\otimes\!\qmn{\blx}\),
\(\knd\!\geq\!3\), \(\rno\in(0,1)\), if  
\(\mQ(\oev)\leq (\sfrac{1}{\sqrt{16\blx}}) e^{-\RD{1}{\vma{\rno}{\mW,\mQ}}{\mQ}-\rno 3\cm{\knd}}\)
for \(\oev\in\outA\) and
\(\cm{\knd}
\DEF\left(\sum\nolimits_{\tin=1}^{\blx} \EXS{\vma{\rno}{\mW,\mQ}}{\abs{\ln \der{\wmn{\tin\sim}}{\qmn{\tin}}
\!-\!\EXS{\vma{\rno}{\mW,\mQ}}{\ln \der{\wmn{\tin\sim}}{\qmn{\tin}}}}^{\knd}}\right)^{\frac{1}{\knd}}\)
then
\(\mW(\outS\setminus\oev)\geq  (\sfrac{1}{\sqrt{16\blx}}) e^{-\RD{1}{\vma{\rno}{\mW,\mQ}}{\mW}-(1-\rno)3\cm{\knd}}.\)
\end{lemma}
Lemma \!\ref{lem:HT} is in the spirit of \cite[\!Thm \!\!5]{shannonGB67A}, but instead of 
Chebyshev ineq, it relies on Berry-Essen Thm via \cite[\!Lem \!\!\ref*{B-lem:berryesseenN}]{nakiboglu16B}.

Our sphere packing bounds are expressed in terms of the averaged Augustin capacity 
and\footnote{Note \(\CRCI{\rno}{\Wm}{\cc}{\epsilon}=\CRCI{\rno}{\Wm}{\cset(\cc)}{\epsilon}\) and 
\(\spa{\epsilon}{\rate,\Wm,\cc}=\spa{\epsilon}{\rate,\Wm,\cset(\cc)}\).}
averaged sphere packing exponent: for all \(\epsilon\in(0,1)\) and  \(\rate\in\reals{+}\):
\begin{align}
%\label{eq:def:avcapacity}
\notag
\CRCI{\rno}{\Wm}{\cset}{\epsilon}
&\DEF \tfrac{1}{\epsilon}\int_{\rno-\epsilon\rno}^{\rno+\epsilon(1-\rno)} 
\left[ 1\vee \left(\tfrac{\rno}{1-\rno}\tfrac{1-\rnf}{\rnf}\right)\right] \CRC{\rnf}{\Wm}{\cset}  \dif{\rnf}
\\
%\label{eq:def:avspherepacking}
\notag
\spa{\epsilon}{\rate,\Wm,\cset}
&\DEF \sup\nolimits_{\rno\in (0,1)} \tfrac{1-\rno}{\rno} \left(\CRCI{\rno}{\Wm}{\cset}{\epsilon}-\rate\right).
\end{align}
\begin{lemma}\label{lem:avspherepacking}
For any \(\rno\!\in\!(0,1]\), \(\Wm\!:\!\inpS\!\to\!\pmea{\outA}\), \(\cset\subset\pdis{\inpS}\) s.t.
\(\CRC{\sfrac{1}{2}}{\Wm}{\cset}\in \reals{+}\), \(\rnf\in (0,1)\),
\(\rate\in [\CRC{\rnf}{\Wm}{\cset},\CRC{1}{\Wm}{\cset})\) and  \(\epsilon\in(0,\rnf)\).
Then  \(0\leq \spa{\epsilon}{\rate,\Wm,\cset}-\spe{\rate,\Wm,\cset}
\leq \tfrac{\epsilon}{\rnf-\epsilon} \tfrac{\rate}{\rnf}\).
\end{lemma}
Proof of Lemma \ref{lem:avspherepacking} is identical to that of 
\cite[Lem \ref*{B-lem:avspherepacking}]{nakiboglu16B}.
\clearpage

\begin{lemma}\label{lem:CCaugustin}
For any product channel \(\Wm\) for the index set \(\{1,\ldots,\blx\}\),
cost function \(\cf\) satisfying \(\cf(\dinp)=\sum_{\tin\in\tinS}\cf_{\tin}(\dinp_{\tin})\) 
for \(\cf_{\tin}:\inpS_{\tin}\to\reals{\geq0}^{\lnd}\),
\(\cc\in \inte{\fcc{\cf}}\), and integers \(M\), \(L\) satisfying 
\begin{align}
%\label{eq:lem:CCaugustin:gamma}
\notag
\tfrac{M}{L}
&>\tfrac{8 e^{2}(1-\rno_{0})(1-\epsilon_{1})\epsilon_{2} \blx^{2.5} }{\epsilon_{1} (1-\epsilon_{2})}
e^{\CRCI{\rno_{0}}{\Wm}{\cc}{\epsilon_{1}}+\frac{\gamma}{1-\rno_{0}}}
\\
\notag
\gamma
&\DEF 3\sqrt[\knd]{3}\left(\sum\nolimits_{\tin=1}^{\blx} 
\left((\CRC{\sfrac{1}{2}}{\Wmn{\tin}}{\cc}+\ln \tfrac{1}{\epsilon_{2}})\vee \knd \right)^{\knd} \right)^{\frac{1}{\knd}}
\end{align}
for a \(\knd\geq 3\), an \(\rno_{0}\in(0,1)\),  an \(\epsilon_{1}\in(0,1)\) 
and an \(\epsilon_{2}\in(0,1)\) 
satisfying  \(\tfrac{(\blx-1)(1-\rno_{0})(1-\epsilon_{1})}{\epsilon_{1}}\geq 1\),
any \((M,L)\) channel code \((\enc,\dec)\) on \(\Wm\) 
satisfying \(\vee_{\dmes\in\mesS}\cf(\enc(\dmes))\leq\cc\) satisfies
\begin{align}
\notag
%\label{eq:lem:CCaugustin}
\Pem{av}
&\geq \left(\tfrac{\epsilon_{1} e^{-2\gamma}}{8 e^{2}(1-\rno_{0})(1-\epsilon_{1}) \blx^{1.5}}\right)^{\frac{1}{\rno_{0}}}  
e^{-\spa{\epsilon_{1}}{\ln \frac{M}{L},\Wm,\cc}}.
\end{align}
\end{lemma} 
\begin{proof}[Proof Sketch]
Since \(\cc\in\inte{\fcc{\cf}}\), 
\(\forall\rno\in(0,1)\exists\lm_{\rno,\!\Wm\!,\cc}\!\in\!\reals{\geq0}^{\lnd}\) s.t.
\(\CRC{\rno}{\!\Wm\!}{\cc}\!=\!\RCL{\rno}{\Wm}{\lm_{\rno,\!\Wm\!,\cc}}\!+\!\lm_{\rno,\!\Wm\!,\cc}\!\cdot\!\cc\)
by  Lem.\ref{lem:Lcapacity}-(\ref{Lcapacity:interior}).
Then \(\qmn{\rno,\!\Wm\!,\cc}\!=\!\qma{\rno,\!\Wm\!}{\lm_{\rno\!,\!\Wm\!,\cc}\!}\) by Lemma \!\ref{lem:Lcenter}.
Furthermore, 
\(\!\qma{\rno,\Wm}{\lm_{\rno,\!\Wm\!,\cc}}\!=\!\oprod{\tin}\!\qma{\rno,\Wmn{\tin}}{\lm_{\rno,\!\Wm\!,\cc}}\)
by Lemma \!\ref{lem:Lcapacityproduct}.
Then \(\qma{\!\rno,\!\Wmn{\tin}}{\lm_{\rno,\!\Wm\!,\cc}}\!:\!(0,1)\!\to\!\pmea{\!\outA_{\tin}\!}\) is continuous 
in \(\rno\) for the total variation topology on \(\pmea{\!\outA_{\tin}\!}\)
because \(\qmn{\rno,\!\Wm,\cc}\) is by Lemma \ref{lem:centercontinuity}.
Then \(\qma{\cdot,\Wmn{\tin}}{\lm_{\cdot,\cc}}\) is a transition probability from 
\(((0,1),\rborel{(0,1)})\) to \((\outS_{\tin},\outA_{\tin})\). 
We define \(\qma{\rno,\Wmn{\tin}}{\epsilon}\) as the \(\outS_{\tin}\) marginal of 
the probability measure  \(\mU_{\rno,\epsilon}\circ\qma{\cdot,\Wmn{\tin}}{\lm_{\cdot,\cc}}\) 
where \(\mU_{\rno,\epsilon}\) is the uniform probability distribution on  \((\rno-\epsilon\rno,\rno+\epsilon(1-\rno))\):
\begin{align}
\label{eq:def:avncenter}
\qma{\rno,\Wmn{\tin}}{\epsilon}
&=\tfrac{1}{\epsilon} \int_{\rno-\rno\epsilon}^{\rno+(1-\rno)\epsilon} \qma{\rnf,\Wmn{\tin}}{\lm_{\rnf,\!\Wm\!,\cc}} \dif{\rnf}.
\end{align}
Let \(\enc_{\tin}(\dmes)\) be the \(\outS_{\tin}\) marginal of \(\enc(\dmes)\)
and \(\qmn{\rno,\tin}\), \!\(\qmn{\rno}\), \!\(\vma{\rno}{\dmes}\) \!be
\begin{align}
%\label{eq:CCaugustin-1}
\notag
\qmn{\rno,\tin}
&\!\DEF\!(1\!\!-\!\!\epsilon_{2})\qma{\rno,\Wmn{\tin}}{\epsilon_{1}}
\!\!+\!\epsilon_{2}\qmn{\frac{1}{2},\Wmn{\tin},\cc}
&
\qmn{\rno}
&\!\DEF\!\oprod{\tin}\!\!\qmn{\rno,\tin}
&
\vma{\rno}{\dmes} 
&\!\DEF\!\vma{\rno}{\enc(\dmes),\qmn{\rno}}.
\end{align}
By \cite[Lem \ref*{A-lem:divergenceQ}-(\ref*{A-divergenceQ-RM},\ref*{A-divergenceQ-Qconvexity})]{nakiboglu16A},
Lemma \ref{lem:Lcapacityproduct} and \(\ln \tau\leq \tau-1\) we have
\begin{align}
\notag
\RD{\rno}{\enc(\dmes)}{\qmn{\rno}}
&\!\leq\!\tfrac{\blx \epsilon_{2}}{1-\epsilon_{2}}
+\int_{\rno(1-\epsilon_{1})}^{\rno+(1-\rno)\epsilon_{1}} 
\tfrac{\RD{\rno}{\enc(\dmes)}{\qma{\rnf,\Wm}{\lm_{\rnf,\Wm,\cc}}}}{\epsilon_{1}} \dif{\rnf}.
\end{align}
Using Lemma \ref{lem:Lcenter}, \cite[Lem \ref*{A-lem:divergenceQ}-(\ref*{A-divergenceQ-order})]{nakiboglu16A}, 
\cite[Prop 2]{ervenH14}, Theorem \ref{thm:minimax}, \(\cf(\!\enc(\!\dmes\!)\!)\!\leq\!\cc\)
and the definition of \(\CRCI{\rno}{\Wm}{\cset}{\epsilon}\) we get
\begin{align}
\label{eq:CCaugustin-2}
\RD{\rno}{\enc(\dmes)}{\qmn{\rno}}
&\leq \tfrac{\blx \epsilon_{2}}{1-\epsilon_{2}}+ \CRCI{\rno}{\Wm}{\cc}{\epsilon_{1}}.
\end{align}
Let \((\enc_{\tin}(\dmes))_{\sim}\) be the component of \(\enc_{\tin}(\dmes)\) 
that is absolutely continuous in \(\qmn{\rno,\tin}\).
Furthermore, let
\(\cla{\rno,\tin}{\dmes}\) and \(\cla{\rno}{\dmes}\) be 
\begin{align}
\notag
\cla{\rno,\tin}{\dmes}
&\DEF\!\ln\!\der{(\enc_{\tin}(\dmes))_{\sim}}{\qmn{\rno,\tin}}
\!-\!\EXS{\vma{\rno}{\dmes}}{\!\ln \der{(\enc_{\tin}(\dmes))_{\sim}}{\qmn{\rno,\tin}}\!}
&&&
\cla{\rno}{\dmes}
&\!\DEF\!\sum\nolimits_{\tin=1}^{\blx}\!\cla{\rno,\tin}{\dmes}.
\end{align}
Then using 
\cite[Lem \ref*{B-lem:MomentBound}]{nakiboglu16B},
\cite[Lem \ref*{A-lem:divergenceQ}-(\ref*{A-divergenceQ-RM},\ref*{A-divergenceQ-order})]{nakiboglu16A},
\cite[Prop 2]{ervenH14},
and
Theorem \ref{thm:minimax} we get 
\begin{align}
\notag
\EXS{\vma{\rno}{\dmes}}{\abs{\cla{\rno,\tin}{\dmes}}^{\knd}}^{\frac{1}{\knd}}
&\leq 3^{\frac{1}{\knd}} \tfrac{(\CRC{\sfrac{1}{2}}{\Wmn{\tin}}{\cc}-\ln\epsilon_{2})\vee \knd}{\rno (1-\rno)}
&
&\forall\knd\in\reals{+},\rno\in(0,1).
\end{align}
Then using the definition of \(\gamma\)  we get
\begin{align}
\label{eq:CCaugustin-3}
\left[\sum\nolimits_{\tin=1}^{\blx} \EXS{\vma{\rno}{\dmes}}{\abs{\cla{\rno,\tin}{\dmes}}^{\knd}}\right]^{\frac{1}{\knd}}
&\leq \tfrac{\gamma}{3\rno(1-\rno)}.
\end{align}
On the other hand, \(\forall\!\dmes\!\in\!\mesS\!,\rno\!\in\!(0,1)\!\) by the definition of \(\vma{\rno}{\dmes}\)
\begin{align}
\label{eq:CCaugustin-4}
\RD{1}{\vma{\rno}{\dmes}}{\qmn{\rno}}
&=\RD{\rno}{\enc(\dmes)}{\qmn{\rno}}-\tfrac{\rno}{1-\rno}\RD{1}{\vma{\rno}{\dmes}}{\enc(\dmes)}.
\end{align}
Thus we can bound \(\RD{1}{\vma{\rno}{\dmes}}{\qmn{\rno}}\) using the non-negativity of the Renyi divergence, i.e. \cite[Thm 8]{ervenH14},
and equation (\ref{eq:CCaugustin-2}) as
\(\RD{1}{\vma{\rno}{\dmes}}{\qmn{\rno}}\leq \tfrac{\blx\epsilon_{2}}{1-\epsilon_{2}}+ \CRCI{\rno}{\Wm}{\cc}{\epsilon_{1}}\). Hence,
\begin{align}
\notag
\lim_{\rno \to \rno_{0}} 
\RD{1}{\vma{\rno}{\dmes}}{\qmn{\rno}}+\tfrac{\gamma}{3(1-\rno)}
&<\ln \tfrac{M}{L}+\ln \tfrac{\epsilon_{1}}{8e^{2}(1-\rno_{0})(1-\epsilon_{1}) \blx^{1.5}}
\\
\notag
\lim_{\rno \to 1} 
\RD{1}{\vma{\rno}{\dmes}}{\qmn{\rno}}+\tfrac{\gamma}{3(1-\rno)}
&=\infty.
\end{align}
\(\RD{1}{\vma{\rno}{\dmes}}{\qmn{\rno}}\) is continuous in \(\rno\) by
\cite[Lem \ref*{B-lem:tilting}]{nakiboglu16B}, then
by the intermediate value theorem \cite[\!4.23]{rudin}
\(\!\forall\!\dmes\!\in\!\mesS\exists\rno_{\dmes}\!\in\!(\rno_{0},1)\) \!s.\!t.
\begin{align}
%\label{eq:CCaugustin-5}
\notag
\RD{1}{\vma{\rno}{\dmes}}{\qmn{\rno}}+\tfrac{\gamma}{3(1-\rno)}\vert_{\rno=\rno_{\dmes}} 
&=\ln \tfrac{M}{L}+\ln \tfrac{\epsilon_{1}}{8 e^{2}(1-\rno_{0})(1-\epsilon_{1}) \blx^{1.5}}.
\end{align}
Lemma \ref{lem:CCaugustin} follows from Lemma \ref{lem:HT} through a pigeon hole argument
similar to the one invoked in \cite[eq (\ref*{B-eq:augustinM-9})-(\ref*{B-eq:augustinM-10})]{nakiboglu16B}.
\end{proof}
If \(\Wm\) is stationary and memoryless  Lemma \!\ref{lem:CCaugustin} can be proved 
\(\forall\!\cc\!\in\!\fcc{\cf}\)
by setting \(\qma{\rno,\Wmn{\tin}}{\epsilon}=\int \mU_{\rno,\epsilon}\circ\qmn{\rnf,\Wmn{\tin},\frac{\cc}{\blx}} \dif{\rnf}\).
Furthermore, bound given in (\ref{eq:lem:Saugustin}) can be obtained for codes satisfying a convex empirical 
distribution constraint \(\cset\subset\pdis{\inpS_{1}}\) by setting
\(\qma{\rno,\Wmn{\tin}}{\epsilon}\!\!=\!\!\!\int\!\!\mU_{\rno,\epsilon}\!\circ\!\qmn{\rnf,\!\Wmn{\tin},\!\cset}\!\dif{\rnf}\)
and \(\qmn{\rno,\tin}
\!=\!(\!1\!-\!\epsilon_{2}\!)\!\qma{\rno,\!\Wmn{\tin}}{\epsilon_{1}}\!+\!\epsilon_{2}\qmn{\frac{1}{2},\Wmn{\tin},\set{B}_{\cset}}\)
where \(\set{B}_{\cset}\DEF\pdis{\{\dinp\in\inpS_{1}:\exists\mP\in\cset\mbox{~s.t.~}\mP(\dinp)>0\}}\).
\begin{align}
\notag
\tilde{\gamma}
&=3\sqrt[\knd]{3\blx}
\left((\CRC{\sfrac{1}{2}}{\Wmn{1}}{\set{B}(\cset)}+\ln \tfrac{1}{\epsilon_{2}})\vee \knd \right)
\\
\label{eq:lem:Saugustin}
\Pem{av}
&\geq \left(\tfrac{\epsilon_{1} e^{-2\tilde{\gamma}}}{8 e^{2}(1-\rno_{0})(1-\epsilon_{1}) \blx^{1.5}}\right)^{\frac{1}{\rno_{0}}}  
e^{-\blx\spa{\epsilon_{1}}{\frac{1}{\blx}\ln \frac{M}{L},\Wmn{1},\cset}}.
\end{align}
\section*{Acknowledgment}
Author would like to thank Fatma Nakibo\u{g}lu and Mehmet Nakibo\u{g}lu for their hospitality; 
this work simply would not have been possible without it.  

%\bibliographystyle{plain} 
%\bibliographystyle{IEEEtran} 
%\bibliography{main}
%\newcommand{\noopsort}[1]{} \newcommand{\printfirst}[2]{#1}
%\newcommand{\singleletter}[1]{#1} \newcommand{\switchargs}[2]{#2#1}

\end{document}